\documentclass{journalA4}
\usepackage{color, graphicx,amssymb,amsmath,wrapfig}%,amsthm
\usepackage[ruled, noline, algo2e, noend, lined,boxed,linesnumbered]{algorithm2e}
\usepackage{microtype}
\usepackage{authblk}

\graphicspath{{figures/}}
\newcommand{\etal}{\emph{et al.}\xspace}
\renewcommand{\epsilon}{\varepsilon}

\newcommand{\BeginMyItemize}{\begin{itemize}\setlength{\itemsep}{-\parskip}}
\newcommand{\EndMyItemize}{\end{itemize}}
\newcommand{\myitemize}[1]{\BeginMyItemize #1 \EndMyItemize}
\newcommand{\BeginMyEnumerate}{\begin{enumerate}\setlength{\itemsep}{-\parskip}}
\newcommand{\EndMyEnumerate}{\end{enumerate}}
\newcommand{\myEnumerate}[1]{\BeginMyEnumerate #1 \EndMyEnumerate}

\newcommand{\BeginMyDescription}{\begin{description}\setlength{\itemsep}{-\parskip}}
\newcommand{\EndMyDescription}{\end{description}}

\newenvironment{lem-without-number}[1]{\paragraph{{\rm\bf Lemma~{#1}}} \sl}{\\[2mm]}   % for repititions of lemmas in an appendix
\newenvironment{thm-without-number}[1]{\paragraph{{\rm\bf Theorem~{#1}}} \sl}{\\[2mm]}   % for repititions of lemmas in an appendix
\newcommand{\mycomment}[2]{\footnote{\textcolor{blue}{#1 #2}}}
\newif\ifComments
\Commentsfalse
\ifComments
  \newcommand{\markdb}[1]{\mycomment{Mark:}{#1}}
  \newcommand{\marcel}[1]{\mycomment{Marcel:}{#1}}
  \newcommand{\bettina}[1]{\mycomment{Bettina:}{#1}}
\else
  \newcommand{\markdb}[1]{}
  \newcommand{\marcel}[1]{}
  \newcommand{\bettina}[1]{}
\fi

\DeclareMathSymbol{\Delta}{\mathalpha}{operators}{"01}
\DeclareMathSymbol{\Theta}{\mathalpha}{operators}{"02}
\DeclareMathSymbol{\Omega}{\mathalpha}{operators}{"0A}

\newcommand{\dist}{\mathrm{dist}}
\newcommand{\area}{\mathrm{area}}

\newcommand{\pol}{P}        %polygon of input $\subdiv$

\newcommand{\reg}{R}

\newcommand{\nw}{\mathrm{nw}}
\newcommand{\NE}{\mathrm{ne}}
\newcommand{\sw}{\mathrm{sw}}
\newcommand{\se}{\mathrm{se}}
\newcommand{\north}{\mathrm{n}}
\newcommand{\south}{\mathrm{s}}
\newcommand{\west}{\mathrm{w}}
\newcommand{\east}{\mathrm{e}}
\newcommand{\hor}{\mathrm{h}}
\newcommand{\ver}{\mathrm{v}}
\newcommand{\subdiv}{\mathcal{S}}

\newcommand{\medaxis}{\mathcal{M}} %Medial axis
\newcommand{\length}{\mathrm{length}}

\newcommand{\slope}{\sigma}
\newcommand{\trian}{\mathrm{T}}
\newcommand{\halfbox}{P^-}   %box with half the edge length of $P$
\newcommand{\eps}{\varepsilon}
\newcommand{\triset}{\mathcal{T}} %triangulation of $P$
\newcommand{\myconst}{\alpha}   % 1 / regconst

\newcommand{\qt}{\mathcal{QT}} % quadtree
\newcommand{\Reals}{{\Bbb R}}

\newcommand{\dtb}[1]{\Delta_{#1}}  % distance to boundary of the argument

\newcommand{\pl}{\mathcal{PL}} % general point location structure
\author[1]{Boris Aronov}
\author[2]{Mark de Berg}
\author[3]{David Eppstein}
\author[4,5]{Marcel Roeloffzen}
\author[2]{Bettina Speckmann}

\affil[1]{Dept.\ of Computer Science and Engineering, Tandon School of Engineering, \mbox{New York University, USA}, boris.aronov@nyu.edu}
\affil[2]{Dept.\ of Computer Science, TU Eindhoven, the Netherlands, \{mdberg,speckman\}@win.tue.nl}
\affil[3]{Dept.\ of Computer Science, Donald Bren School of Information and Computer Sciences, \mbox{University of California, Irvine}, eppstein@ics.uci.edu}
\affil[4]{National Institute of Informatics (NII), Tokyo, Japan, marcel@nii.ac.jp}
\affil[5]{JST Kawarabayashi ERATO Large Graph Project}

\title{Distance-Sensitive Planar Point Location\thanks
{
    B.~Aronov has been supported by U.S.-Israel Binational Science Foundation grant~2006/194, by NSF Grants CCF-08-30691, CCF-11-17336, and CCF-12-18791, and by NSA MSP Grant H98230-10-1-0210.
    D.~Eppstein has been supported by NSF grant 1217322 and ONR grant N00014-08-1-1015.
    M.~Roeloffzen and
    B.~Speckmann were supported by the Netherlands' Organisation for
    Scientific Research (NWO) under project no.~600.065.120 and~639.023.208, respectively.}}

\begin{document}

\maketitle

\begin{abstract}
Let $\subdiv$ be a connected planar polygonal subdivision with $n$ edges that we want to preprocess for
point-location queries, and where we are given the probability $\gamma_i$ that the query point lies in a polygon $P_i$ of $\subdiv$.
We show how to preprocess $\subdiv$ such that the query time for a point~$p\in P_i$ depends
on~$\gamma_i$ and, in addition, on the distance from $p$ to the boundary
of~$P_i$---the further away from the boundary, the faster the query.
More precisely, we show that a point-location query can be answered in time
$O\left(\min \left(\log n, 1 + \log \frac{\area(P_i)}{\gamma_i \dtb{p}^2}\right)\right)$,
where $\dtb{p}$ is the shortest Euclidean distance of the query point~$p$ to the boundary of $P_i$.
Our structure uses $O(n)$ space and $O(n \log n)$ preprocessing time.
It is based on a decomposition of the regions of $\subdiv$ into convex quadrilaterals and triangles with the following property: for any point $p\in P_i$, the quadrilateral or triangle containing~$p$ has area $\Omega(\dtb{p}^2)$.
For the special case where $\subdiv$ is a subdivision of the unit square and $\gamma_i=\area(P_i)$,
we present a simpler solution that achieves a query time of $O\left(\min \left(\log n, \log \frac{1}{\dtb{p}^2}\right)\right)$.
The latter solution can be extended to convex subdivisions in three dimensions.

\end{abstract}

\section{Introduction}
Point location is one of the most fundamental problems in computational geometry. Given a subdivision $\subdiv$ the goal is to preprocess it so that we can determine efficiently which region of~$\subdiv$ contains a query point~$p$. Many different variants of the point-location problem exist; in our work we first focus on planar point location in polygonal subdivisions and later extend one of our results to convex polyhedral subdivisions in three dimensions.
In the following, unless otherwise specified, our subdivision~$\subdiv$ is subdivision of a polygonal domain in the plane into polygons. The subdivision need not be conforming---we may have T-junctions, for instance---but when considering a polygon~$P_i$ of the subdivision we ignore the subdivision vertices whose angle inside  $P_i$ is exactly~$\pi$. A triangulation is a subdivision consisting of triangles.

%Related work
There are several different solutions for planar point location that are worst-case optimal. In particular, there are structures that require $O(n\log n)$ preprocessing, use $O(n)$ space, and can answer a point-location query
in $O(\log n)$ time; see the surveys by Preparata~\cite{p1990} and Snoeyink~\cite{s2004} for an overview. In three dimensions no point location structure is known for general subdivisions that uses linear space and has logarithmic query time. For convex subdivisions Preparata and Tamassia~\cite{pt1992} showed that combining dynamic planar point location and persistency techniques yield an $O(n \log^2 n)$ space structure that answers queries in $O(\log^2 n)$ time. This method was later extended and improved so that it works for general subdivisions and requires only $O(n \log n)$ space and preprocessing time for $O(\log^2 n)$ query time~\cite{gt1998,s2004}.

For planar point location a query time of $O(\log n)$ is optimal in the worst case, but it may be possible to do better for certain types of query points. For example, if the query points are not distributed uniformly among the regions of~$\subdiv$,
then it may be desirable to reduce the query time for points in frequently queried regions.
Iacono~\cite{i2001} showed that this is indeed possible: given a triangulation $\subdiv$
where each triangular region $R_i$ has a probability $\gamma_i$ associated with
it---the probability that the query point $p$ falls in $R_i$---then one can
answer a point-location query in expected time $O(H(\subdiv))$, where
\[
H(\subdiv) := \sum_{R_i \in \subdiv} \gamma_i \log \frac{1}{\gamma_i},
\]
is the \emph{entropy} of~$\subdiv$. This result is optimal, because the entropy
is a lower bound on the expected query time~\cite{k1998,s1948}. Several other point-location
structures have been proposed that answer queries in $O(H(\subdiv))$ expected time~\cite{amm2007,ammw2007}.
The structure presented by Arya, Malamatos, and Mount~\cite{amm2007} is relatively simple
and efficient in practice. It works for subdivisions with constant-complexity regions and, for any region $R_i$ the worst-case query time for points inside $R_i$ is $O(1+\min(\log \frac{1}{\gamma_i},\log n))$.
The results mentioned so far assume that the distribution is known in advance.
Recently Iacono~\cite{i2011} proposed an algorithm that eventually achieves $O(H(\subdiv))$ query time,
but does not need any knowledge of the query distribution. Instead, the algorithm changes the structure
according to the queries received. The results mentioned above require the regions of
the input subdivision~$\subdiv$ to have constant complexity. This requirement is
necessary. Indeed, if a subdivision with $n$ edges has only two regions, each with associated probability~1/2,
then we cannot hope to achieve~$O(1)$ query time. One could of course subdivide the regions into constant complexity regions, say triangles, and distribute the query probability evenly among these regions. However, in many cases one would expect that queries are not evenly distributed within each polygon. For example if queries come from users selecting polygons by clicking on them, one would expect most queries to occur far from the boundary as users are inclined to click in the `middle' of a region. This raises the question if it is possible to improve query times depending on where the query point is within the polygon that contains it. In our work we investigate the possibility of relating the query time to the distance of a query point to the nearest point on the boundary of the region that contains it. We call this \emph{distance-sensitive point location}.

Differentiating between query points within higher complexity polygons is not new. Collette~\etal~\cite{cdilm2012} showed how to compute, for any simple polygon~$\pol$ and any probability distribution over~$\pol$, a Steiner triangulation with near-optimal entropy, and they proved that the minimum entropy of any triangulation is a lower bound on the expected query time for point-location in the linear decision-tree model. By applying their Steiner triangulation to every region in the given subdivision, and using the resulting triangles as input for an entropy-based point-location structure, near-optimal expected query time is achieved. In the case of distance-sensitive point location we could define a probability distribution based on the distance of points to the region boundary and construct such a Steiner triangulation. Unfortunately, a near-optimal entropy does not imply any bounds on specific query points. Indeed the construction by Collette~\etal can generate very small triangles, even in high probability areas. A point $p$ that is far from the boundary can end up in such a very small triangle, which has a small total probability. As a result a query for $p$ has a long query time. We will focus on creating a point location structure that guarantees fast query time for any point far from the region boundary.

\paragraph{Problem definition} Let $p$ be a query point inside polygon $\pol_i \in \subdiv$ with area $\area(\pol_i)$ and probability $\gamma_i$ that a query point is inside $\pol_i$. We want the time of a query for $p$ to be
\[
O\left(\min \left(\log n, 1 + \log \frac{\area(P_i)}{\gamma_i\dtb{p}^2}\right)\right).
\]
Here, $\dtb{p}$ denotes the minimum Euclidean distance from $p$ to the boundary of $P_i$. When the polygons in the subdivision have constant complexity, then this can be achieved using, for instance, the entropy-based point-location structure of Arya, Malamatos, and Mount~\cite{amm2007}. Since for any point $p\in P_i$ the distance to the boundary of $P_i$ is $O(\sqrt{\area(P_i)})$, this gives the desired query bound. When polygons have higher complexity we can also use an entropy-based structure, but first have to decompose each polygon into constant complexity regions and assign probabilities appropriately. Specifically we show that it is sufficient to compute for each polygon $\pol \in \subdiv$ with $n_P$ vertices a \emph{distance-sensitive decomposition} of $\pol$ into $O(n_P)$ regions with the following properties:
\myitemize{
\item each region $\reg$ is convex and has constant complexity;
\item for some absolute constant $\myconst$ the decomposition has the \emph{$\myconst$-distance property}:
      for any point $p\in \pol$, the region $\reg$ containing~$p$ has area at
      least $\myconst \cdot \dtb{p}^2$, where $\dtb{p}$ is the
      distance from $p$ to the boundary of $P$.
}
\begin{wrapfigure}[5]{r}{1.5cm}
%\vspace{-30pt}
\raggedleft
\hspace{-10pt}
\includegraphics{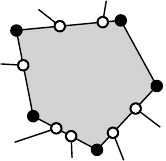}
\end{wrapfigure}

The entropy-based search structure by Arya, Malamatos, and Mount~\cite{amm2007}, which will serve as the backbone of our algorithm, requires its input regions to have constant complexity. Here the complexity of a region is counted as explained earlier: if the interior angle within a region $P_i$ at a subdivision vertex is exactly $\pi$, then that vertex does not count towards the complexity of~$P_i$. For instance, the shaded region in the figure on the right has only five vertices.

The problem of computing a decomposition with these properties can be considered a mesh-generation problem. Many different types of meshes exist;
see the survey by Bern~\cite{b2004} for an overview.
In several of these meshes the number of mesh elements is linear in the complexity of the polygon, and the mesh elements are ``well-shaped''. For example, the meshing algorithm proposed by
Bern~\etal~\cite{bmr1994} produces triangles with angles of at most $90^\circ$.
There are also meshes that are designed to be more
detailed near the polygon boundary and coarser further away from the boundary.
These meshes, however, do not guarantee a relation between the distance to the boundary and the size of mesh elements~\cite{bmr1994,r1995} or they do not have a bound on the number of mesh elements~\cite{beg1994}.
To the best of our knowledge no published mesh generation method guarantees that the mesh consists of $O(n_P)$ elements that have the required distance property.

Another ``query-sensitive'' subdivision was defined by Mitchel, Mount, and Suri to faciltate ray shooting~\cite{mms1997}.  It has the property that the cost of shooting a ray (i.e., walking along it through the subdivision, from its origin until the first point of intersection with the obstacle) is proportional to its ``cover complexity,'' which, roughly speaking, is the minimum number of disks that are required to cover this portion of the ray, with each disk not intersecting ``too much'' of an obstacle.  In our application, the role of the obstacle is taken by the exterior of the region to be subdivided.  However, the structure of~\cite{mms1997} does not seem have the right properties for our purposes.

\paragraph{Our results}
We start by describing in more detail how a \emph{distance-sensitive decomposition} can be used to construct a distance-sensitive point-location structure. We then continue by giving algorithms to compute distance-sensitive decompositions. For convex polygons we actually do not need to use non-conforming subdivisions: we show that any convex polygon can be \emph{triangulated} in such a way that the resulting triangulation has the $\myconst$-distance property for $\myconst=1$.
For possibly non-convex simple polygons we investigate several different settings that have different restrictions on the resulting decomposition.
We show that it is not always possible to create a conforming triangulation with the $\myconst$-distance property without using Steiner points, and that the number of Steiner points needed in such a triangulation cannot be bounded as a function of the complexity of the polygon $\pol$.

Instead, we show that any simple polygon~$P$ can be decomposed into $O(n_P)$ non-conforming convex quadrilaterals and triangles that have the $\myconst$-distance property for some absolute constant $\myconst>0$. The decomposition can be computed in $O(n_P \log n_P)$ time. This result is used to obtain a linear-size data structure for point location in a planar connected polygonal subdivision~$\subdiv$, such that the query time for a point $p$ in a polygon $P_i \in \subdiv$ is $O\left(\min \left(\log n, 1 + \log \frac{\area(P_i)}{\gamma_i\dtb{p}^2}\right)\right)$, where $\dtb{p}$ is the distance from the query point~$p$ to the boundary of its containing region.

Lastly we investigate a special case in which the query bound is based only on the distance of a query point to the boundary. Specifically, assuming the subdivision is contained in a square of area~1, we present a data structure that achieves a query time of $O\left(\min \left(\log n, 1 + \log \frac{1}{\dtb{p}^2}\right)\right)$ for a point $p$. The new structure is based on a depth-bounded quadtree and a worst-case optimal point-location structure, both of which can be constructed in $O(n\log n)$ time and $O(n)$ space. The more general structure presented above achieves the same bounds if we choose $\gamma_i = \area(\pol_i)$, but we believe the new structure is much simpler and may be faster in practice. As a bonus, the new structure achieves the more general bound of $O\left(\min \left(\log n, 1 + \log \frac{\area(P_i)}{\gamma_i\dtb{p}^2}\right)\right)$ for any subdivision of the unit square where $\gamma_i = O(\area(\pol_i))$.

The simpler structure also extends to three dimensions. Specifically, given a convex polyhedral subdivision contained in a unit cube with $n$ edges, we show how to construct a distance-sensitive point location structure in $O(n \log^2 n)$ time and $O(n\log n)$ space that answers a query for a point $p$ in $O(1 + \log \frac{1}{\dtb{p}^2})$ time if $\dtb{p} \geq \sqrt{3} / \sqrt[3]{n}$ and $O(\log^2 n)$ time otherwise. Note that the space requirement comes from the worst-case point location structure and not the additional octree structure that allows for distance-sensitive queries.

\section{Distance-sensitive decomposition of simple polygons}
As argued in the introduction we can use entropy-based point location structures to create a distance-sensitive point location structure by first creating a distance-sensitive decomposition for the polygons of the input subdivision. To avoid confusion we use the term \emph{polygon} for polygons of the input subdivision $\subdiv$ and $\emph{region}$ for the regions of the decomposition of a polygon. Recall that we define a distance-sensitive decomposition as follows: Let $P$ be a simple polygon with $n_P$ edges. A distance-sensitive decomposition of $P$ consists of $O(n_P)$ regions with the following properties:
\myitemize{
\item each region $\reg$ is convex and has constant complexity;
\item for some absolute constant $\myconst$ the decomposition has the \emph{$\myconst$-distance property}:
      for any point $p\in \pol$, the region $\reg$ containing~$p$ has area at least $\myconst \cdot \dtb{p}^2$, where $\dtb{p}$ is the Euclidean distance from $p$ to the boundary of $P$.
}
Given a subdivision $\subdiv$ and for each polygon $P_i \in \subdiv$ its distance-sensitive decomposition $\mathcal{P}^\mathrm{dec}_i$ we can assign each region $\reg \in \mathcal{P}^{\mathrm{dec}}_i$ a weight $\gamma_i \cdot \area(\reg) / \area(P_i)$. We then build the entropy-based structure by Arya~\etal~\cite{ammw2007} on the union of the distance-sensitive decompositions of all polygons in $\subdiv$ using the weights for the probability distribution. Now a point $p$ with distance $\dtb{p}$ to the nearest boundary of the subdivision must be contained in a region $\reg$ with weight $\gamma_i \cdot \area(\reg) / \area(P_i) \geq \myconst \cdot \gamma_i \cdot \dtb{p}^2 / \area(P_i)$. It follows that the query time for $p$ is
\[
O\left(\min \left(\log n, 1 + \log \frac{\area(P_i)}{\gamma_i\dtb{p}^2}\right)\right).
\]
So once we have a distance-sensitive decomposition it is easy to construct a distance-sensitive point location structure.

\begin{theorem}\label{thm:decomp-pl}
Let $\subdiv$ be a subdivision, where for each polygon $P_i \in \subdiv$ we are given a distance-sensitive decomposition $\mathcal{P}_i^\mathrm{dec}$. Then we can construct in $O(n\log n)$ expected time a point location for $\subdiv$ such that, for any query point $p$, the query time is $O\left(\min \left(\log n, 1 + \log \frac{\area(P_i)}{\gamma_i\dtb{p}^2}\right)\right)$, where $\dtb{p}$ is the distance from $p$ to the boundary of the region containing~$p$.
\end{theorem}
Note that the expectation in the construction time has nothing to do with 
the probabilities~$\gamma_i$, but it is because Arya~\etal use randomized incremental construction to build their data structure.
Also note that the distance-sensitive decomposition may be non-conforming, that is, the boundary-edges of a region may contain many interior vertices that are not counted towards its complexity. Indeed, since Arya~\etal use randomized incremental insertion of maximal segments to build their structure, it is not a problem if the decomposition is non-conforming. In the remainder of this section we focus on constructing distance-sensitive decompositions, first for convex polygons and then for arbitrary simple polygons.

\subsection{Convex polygons}\label{s:convex}
As a warm-up exercise, we start with the problem of decomposing a convex polygon~$P$ with $n_P$ vertices so that the decomposition has the $\myconst$-distance property for $\myconst = 1$. For this case the decomposition is actually a triangulation.

Our algorithm is quite simple. First we split~$P$ by adding a diagonal
between the vertices defining the diameter of~$P$. We further decompose
each of the two resulting subpolygons using a recursive algorithm, which we describe next.
We call the edges of the input polygon~$P$ \emph{polygon edges} and the edges created by the subdivision process \emph{subdivision edges}. The boundary of each subpolygon we recurse on
consists of one subdivision edge and a convex chain of polygon edges, where the angles between the chain and the subdivision edge are acute. Let $Q$ be such a subpolygon and $e$ the corresponding subdivision edge. We construct the largest area triangle $T$ contained in $Q$ that has $e$ as an edge by finding the vertex~$v$ on the convex chain that is farthest from~$e$.
Because the chain is convex this vertex can be found in $O(\log n_Q)$ time, where $n_Q$ is the number of vertices of $Q$.
%We get the following theorem.
\begin{theorem}
\label{thm:convex-triangulation}
For any convex polygon $P$ with $n_P$ vertices we can compute in $O(n_P \log n_P)$ time a triangulation that has the $1$-distance property.
\end{theorem}

\begin{figure}
%\vspace{-23pt}
\centering
\includegraphics{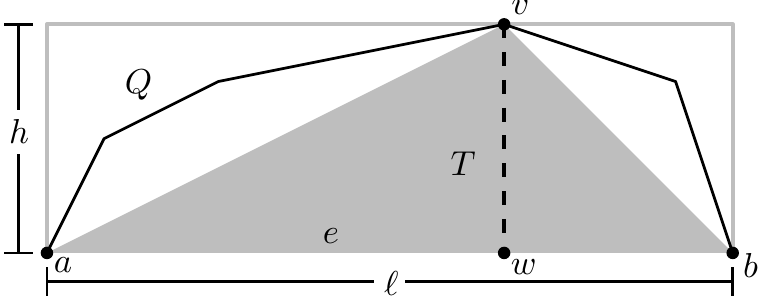}
\caption{A triangle in polygon $Q$ is defined by the subdivision edge $e$ and the point $q$ furthest away from~$e$.}
\label{fig:convex-triangle}
\end{figure}

\noindent
\begin{proof}
Consider the algorithm described above.
We first show that the two angles between the subdivision edge $e$ and the convex chain inside a subpolygon $Q$ are acute by showing that $|e|$ is the diameter of $Q$. For the first two subpolygons, created by cutting the convex polygon across the diagonal the length of the subdivision edge is the diameter of the subpolygon by definition. Now consider a subpolygon $Q$ with subdivision edge $e$ and triangle $T$ formed by $e = (a,b)$ and the furthest point $v \in Q$ from $e$, see Figure~\ref{fig:convex-triangle}. This creates up to two new subpolygons $Q_1$ and $Q_2$ with $e_1 = (a,v)$ and $e_2 = (v,b)$ as subdivision edges. Since the angles between $e$ and the convex chain are acute it follows that $Q$ is contained in a rectangle with side length $e$ and height $h = \dist(v,e)$. To prove that $|e_1|$ and $|e_2|$ are the diameters for $Q_1$ and $Q_2$ respectively consider a point $w \in e$ that is closest to $v$. The edge $(v,w)$ divides the rectangle into two rectangles $R_1$ containing $Q_1$ and $R_2$ containing $Q_2$. The edges $e_1$ and $e_2$ are the diameters of these rectangles and it follows that they are also the diameters of~$Q_1$ and~$Q_2$.

Now consider a subpolygon $Q$ with base edge $e$ and a furthest point $v\in Q$ from $e$. Since $|e|$ is the diameter $Q$, the angles at $e$'s endpoints are acute and $Q$ must be contained in an $\ell \times h$ rectangle where $\ell = |e|$ and $h = \dist(v, e)$. It follows that for any point $p \in T$ we have
\[
\dtb{p}^2 \leq \min(h, \ell/2)^2 \leq h\ell/2 = \area(T).
\]

The diameter of a convex polygon can be computed in $O(n_P)$ time, and the creation of each triangle takes $O(\log n_P)$ time. Since there are $n_P-2$ triangles it follows that the algorithm takes $O(n_P \log n_P)$ time.
\end{proof}

Combining this result with Theorem~\ref{thm:decomp-pl} we obtain the following corollary.

\begin{cor}
Let $\subdiv$ denote a convex planar polygonal subdivision with $O(n)$ vertices and let $\gamma_i$ for each $P_i \in \subdiv$ denote the probability that a query point lies in $P_i$. We can construct in $O(n\log n)$ expected time a point location structure that uses $O(n)$ space and answers a query with a point $p$ in $O\left(\min \left(\log n, 1 + \log \frac{\area(P_i)}{\gamma_i\dtb{p}^2}\right)\right)$ time, where $\dtb{p}$ denotes the Euclidean distance from $p$ to the nearest point on any edge of $\subdiv$.
\end{cor}

\subsection{Arbitrary polygons}\label{ss:arbitrary}
We now consider non-convex polygons. We wish to compute a decomposition
of a simple polygon $P$ into constant-complexity regions that have the
$\myconst$-distance property.
This is not always possible with a triangulation. Consider the polygon in Figure~\ref{fig:many-steiner}, where the width of the middle column coming up from the bottom edge is $\eps$. Any triangulation of the polygon in Figure~\ref{fig:many-steiner}a must include
triangle $uvw$ or $uvz$, and when $\eps$ tends to zero
the $\myconst$-distance property is violated for points in the middle of these triangles.
A Steiner triangulation with the $\myconst$-distance property always
exists: the quadtree-based mesh of Bern~\etal~\cite{beg1994} can be adapted
to have the $\myconst$-distance property---the (small) adaptations are required only around acute angles. However, the number of Steiner points
and, hence, the number of triangles
cannot be bounded as a function of the number of vertices of $P$.
Next we show that this is necessarily so.

\begin{theorem}
\label{th:many-steiner}
For any constant $\myconst>0$ and any $m>0$, there is a simple polygon~$P$ with eight vertices
such that any Steiner triangulation of $P$ with the $\myconst$-distance property
uses at least $m$ Steiner points.
\end{theorem}

\begin{figure}
\centering
\includegraphics{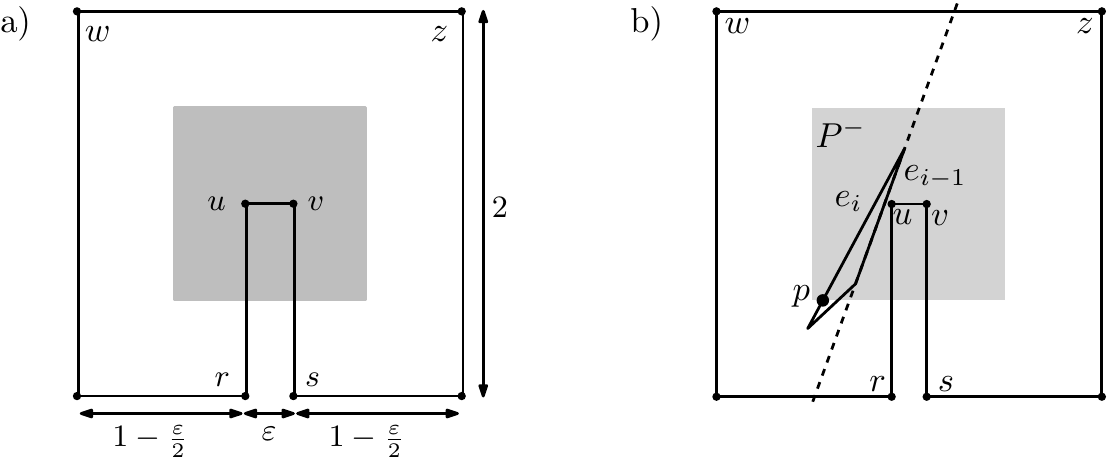}
\caption{a) Any triangulation of $P$ with the $\myconst$-distance property requires many Steiner points. b) Triangle $T_{i-1}$ intersects the boundary of $\halfbox$ in $p$.}
\label{fig:many-steiner}
\end{figure}

\noindent
\begin{proof}
Let $P$ be the polygon shown in Figure~\ref{fig:many-steiner}a.
Consider a Steiner triangulation $\triset$ of $P$ with the $\myconst$-distance property.
Let $T_0$ be the triangle in $\triset$ that has $uv$ as an edge.\footnote{The edge $uv$ can contain Steiner vertices in its interior, as the only requirement we have for the Steiner triangulation is that any two triangles either meet in a single vertex, along a complete edge, or not at all. When $uv$ contains Steiner vertices, we can replace $uv$ by any subedge of $uv$, and the argument still holds.}
If $\eps$ is very small then the other two edges of $T_0$ cannot be very long
either, otherwise the $\myconst$-distance property is violated inside~$T_0$.
This in turn implies that the neighboring triangles of $T_0$ cannot be very large.
The idea is to repeat this argument to show that many triangles are needed
to cover~$P$.

Specifically, we define a sequence
of triangles $T_0, T_1, T_2, \ldots$, as follows.
Suppose we are given a triangle $T_i$ and an edge $e_i$ bounding $T_i$ from below.
(For $i=0$, we have $e_i = uv$.) Consider the other two edges of $T_i$.
We select one of these two edges as $e_{i+1}$ and define
$T_{i+1}$ as the triangle directly above~$e_{i+1}$. We select~$e_{i+1}$
as follows. If only one of the edges bounds $T_i$ from above, then this edge
is selected. If both edges bound $T_i$ from above, then we select
the edge with the smaller absolute slope. This selection guarantees that for every edge~$e_i$ at least one endpoint is above~$e_0$.

Our goal is now to prove that the size of the triangles $T_0,T_1,\ldots$ does not increase
too rapidly---more precisely, that $T_{i+1}$ cannot be arbitrarily larger
than~$T_i$. This requires an invariant on the length of the edges $e_i$,
but also on their absolute slope. We denote the absolute slope of~$e_i$ by~$\sigma_i$.
Thus $\slope_i = |e_i|_y / |e_i|_x$, where $|e_i|_x$ and $|e_i|_y$ denote the lengths of the projection of~$e_i$ on the $x$- and $y$-axis.
Let $\halfbox$ denote the square with edge length~1 centered at the midpoint of~$uv$.
In Figure~\ref{fig:many-steiner}a this square is shaded. Our argument
will use the fact that for $T_i$ inside $\halfbox$ the nearest
boundary point for any $p\in T_i$ lies on $uv$, $ur$, or $vs$.
We show that both the slope and length of edge $e_i$ are bounded as a function of $i$,
and that $e_i$ remains inside $\halfbox$, until $\slope_i \cdot |e_i|$ is large enough.
More precisely, we can prove that as long as
$\max(4, \slope_i^2) \cdot |e_i| < \frac{\myconst}{8\sqrt{2}}$
the following three properties hold, where (i) and (ii) are needed to
prove~(iii):
\myEnumerate{
\item[(i)] edge $e_i$ is contained in $\halfbox$;
\item[(ii)] the slope $\sigma_i$ of $e_i$ satisfies $\sigma_i \leq  (2^{i+1}-2) / \myconst$;
\item[(iii)] edge $e_i$ has length at most $8\eps \cdot 2^{(i+1)(i+7)} / (\myconst^{3i})$.
}
These properties can be proven using induction, where the proof for (ii) requires (i) and the proof for (iii) requires (i) and (ii). It is easy to see that (i), (ii) and (iii) are true for $e_0$ and the step cases are given in Lemmas~\ref{lem:prop-i},~\ref{lem:prop-ii}~and~\ref{lem:prop-iii}, respectively.
It follows from property~(iii) that we can always choose $\eps$ small enough that we need at least $m$ Steiner points before $T_i$ can leave the square~$\halfbox$.
\end{proof}

\begin{lemma}
\label{lem:prop-i}
If for $e_j$ with $0\leq j < i$ we have $\max(4, \slope_j^2) \cdot |e_j| < \frac{\myconst}{8\sqrt{2}}$ and $e_j$ is contained in $\halfbox$, then $e_i$ is contained in $\halfbox$.
\end{lemma}
\begin{proof}
We assume for a contradiction that $e_i$ extends outside of $\halfbox$ and show that if this is the case, then $T_{i-1}$ does not have the $\myconst$-distance property for the given $\myconst$. The area of $T_{i-1}$ is upper bounded by $|e_{i-1}| \cdot |e_{i}| \leq |e_{i-1}| \cdot 2\sqrt{2}$. Since $e_i$ extends outside of $\halfbox$ and $e_{i-1}$ is inside it there must be a point $p \in T_{i-1}$ that is on the boundary of $\halfbox$. If $p$ is on the left, top or right edge of $\halfbox$ then $\dtb{p} \geq (1-\eps)/2 \geq 1/4$. If $p$ is on the bottom edge of $\halfbox$ we can use the slope of $e_{i-1}$ and the fact that its top endpoint is above $e_0$ to bound the distance from $p$ to the boundary of $P$. Without loss of generality assume that $p$ is to the left of $ur$. Since one endpoint of $e_{i-1}$ is above $e_0$ and $e_{i-1}$ cannot intersect $e_0$ the distance from $p$ to $ur$ (the nearest boundary edge) is at least $1/(2\slope_{i-1})$, see also Figure~\ref{fig:many-steiner}b. We get that $\dtb{p} \geq 1/(2\max(2, \slope_{i-1}))$. This would imply that
\begin{align*}
\area(T_i) \leq |e_{i-1}| \cdot 2\sqrt{2}
< \frac{\myconst}{8\sqrt{2}\max(4,\slope_{i-1}^2)} \cdot 2\sqrt{2}
\leq \myconst \cdot \dtb{p}^2,
\end{align*}
contradicting that $T_i$ has the $\myconst$-distance property. Hence, we can conclude that $e_{i}$ must be contained in $\halfbox$.
\end{proof}

\begin{lemma}
\label{lem:prop-ii}
If for all $0 \leq j < i$ the edge $e_j$ is inside $\halfbox$ and $\max(4, \slope_j^2) \cdot |e_j| < \frac{\myconst}{8\sqrt{2}}$, then $\slope_i \leq (2^{i+1} - 2) /\myconst$.
\end{lemma}

\begin{figure}[t]
\centering
\includegraphics{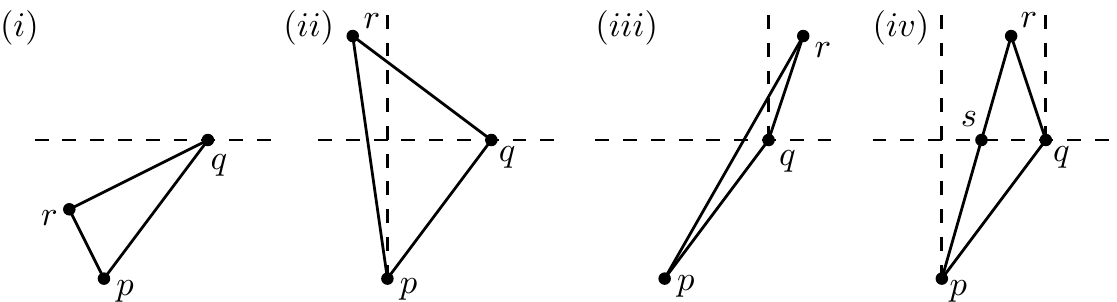}
\caption{The four cases for a point $r$ used in the proof of Lemma~\ref{lem:prop-ii}}
\label{fig:triangulation-cases}
\end{figure}

\begin{proof}
First note that from Lemma~\ref{lem:prop-i} we know that $e_i$ is contained in $\halfbox$.
Let $p$ and $q$ denote the bottom and top endpoints of $e_{i-1}$, and let $r$ denote the third vertex of $T_{i-1}$. Without loss of generality we assume that $p$ is to the left of $q$. We distinguish cases based on the location of $r$ relative to $p$ and $q$ (see Figure~\ref{fig:triangulation-cases}).

\smallskip
\emph{Case (i): $r$ is below $q$.}
If $r$ is below $q$, then $e_{i}$ is the edge $qr$. Since $r$ must be above the supporting line of $pq$ and left of $q$, we get that $\slope_{i} \leq \slope_{i-1}$.

\smallskip
\emph{Case (ii): $r$ is to the left of $p$ and above $q$.}
If $r$ is left of $p$ and above $q$, then $e_i = rq$.
Let $\area(\trian_{i-1})$ denote the area of $\trian_{i-1}$ and $\dist(\trian_{i-1})$ the maximum distance from any point in $\trian_{i-1}$ to the boundary of $P$.
Since $\trian_{i-1}$ is contained within $\halfbox$ and $q$ is above $e_0$, it follows that $\dist(\trian_{i-1}) \geq |rq|_y$. Due to our assumptions on the positions of $p,q,r$ it follows that $\area(\trian_{i-1}) \leq |rp|_y \cdot |rq|_x$.  This allows us to bound the slope $\slope_{i}$ as
    \[
    \slope_{i} - \slope_{i-1}
    = \sigma_{i} \left(1 - \frac{|pq|_y}{|pq|_x} \cdot \frac{|rq|_x}{|rq|_y} \right)
    \leq \sigma_{i} \left(1 - \frac{|pq|_y}{|rq|_y}\right)
    \leq \sigma_{i} \left(1 - \frac{|pq|_y}{|rp|_y}\right)
    \]\[
    = \sigma_{i} \frac{|rp|_y - |pq|_y}{|rp|_y}
    = \sigma_{i} \frac{|rq|_y}{|rp|_y}
    = \frac{|rq|_y^2}{|rq|_x \cdot |rp|_y}
    \leq \frac{\dist(\trian_{i-1})^2}{\area(\trian_{i-1})}
    \leq 1/\myconst.
    \]

\smallskip
\emph{Case (iii): $r$ is to the right of $q$.} As before we have $\dist(\trian_{i-1}) \geq |rq|_y$ and $\area(\trian_{i-1}) \leq |rq|_y \cdot |pq|_x$. We get
    \[
    \slope_{i} - \slope_{i-1}
    = \frac{|rp|_y}{|rp|_x} - \frac{|pq|_y}{|pq|_x}
    \leq \frac{|rp|_y - |pq|_y}{|pq|_x}
    = \frac{|rq|_y}{|pq|_x}
    = \frac{|rq|_y^2}{|rq|_y \cdot |pq|_x}
    \leq 1/\myconst.
    \]

\smallskip
\emph{Case (iv): $r$ is horizontally between $p$ and $q$.} This case provides us with two edges that face upward. Recall that in this case $e_i$ is the edge with smaller slope. We further split this case into three subcases. First assume that $\slope_{rp} \leq \slope_{rq}$. Let $s$ denote a point on $rp$ with the same $y$-coordinate as $q$, so $|sq|_x = |rp|_x \cdot \frac{|rq|_y}{|rp|_y} + |rq|_x$. We bound $\area(\trian_{i-1}) \leq |sq|_x \cdot |rp|_y$ and $\dist(\trian_{i-1}) \geq |rq|_y$. We get
    \[
    \slope_{i} - 2 \slope_{i-1}
    = \slope_{rp} - 2 \slope_{pq}
    = \slope_{rp} \left( 1 - 2 \frac{|pq|_y}{|pq|_x} \cdot \frac{|rp|_x}{|rp|_y} \right)
    \leq\footnotemark \slope_{rp} \left(1 - \frac{|pq|_y}{|rp|_y} \right)
    \]\[
    = \slope_{rp} \frac{|rq|_y}{|rp|_y}
    = 2 \frac{|rq|_y}{\frac{2}{\slope_{rp}} |rp|_y}
    \leq 2 \frac{|rq|_y}{\left(\frac{1}{\slope_{rp}} + \frac{1}{\slope_{rq}}\right) \cdot |rp|_y}
    = 2 \frac {|rq|_y} {|rp|_y \cdot \left( \frac{|rp|_x}{|rp|_y} + \frac{|rq|_x}{|rq|_y}\right)}
    \]\[
    = 2 \frac {|rq|_y^2} {|rp|_y \cdot \left(|rp|_x \cdot \frac{|rq|_y}{|rp|_y} + |rq|_x\right)}
    = 2 \frac{|rq|_y^2}{|rp|_y \cdot |sq|_x}
    \leq 2 \frac{\dist(\trian_{i-1})^2}{\area(\trian_{i-1})}
    \leq 2 / \myconst.
    \]

\footnotetext{This step assumes that $\frac{|rp|_x}{|pq|_x} \geq 1/2$, which follows from our assumption that $\slope_{rp} \leq \slope_{rq}$.}

The second case we assume that $\slope_{rq} < \slope_{rp}$ and $|rp|_x \geq |pq|_x / 2$. This leads to a very similar calculation to the previous case, namely
    \[
    \slope_{i} - 2\slope_{i-1}
    = \slope_{rq} - 2\slope_{pq}
    = \slope_{rq} \left( 1 - 2 \frac{|pq|_y}{|pq|_x} \cdot \frac{1}{\slope_{rq}} \right)
    \leq \slope_{rq} \left( 1 - 2 \frac{|pq|_y}{|pq|_x} \cdot \frac{1}{\slope_{rp}} \right)
    \]\[
    \leq \slope_{rq} \left(1 - \frac{|pq|_y}{|rp|_y} \right)
    = \slope_{rq} \frac{|rq|_y}{|rp|_y}
    = 2 \frac{|rq|_y}{\frac{2}{\slope_{rq}} |rp|_y}
    \leq 2 \frac{|rq|_y}{\left(\frac{1}{\slope_{rp}} + \frac{1}{\slope_{rq}}\right) \cdot |rp|_y}
    \]\[
    = 2 \frac {|rq|_y} {|rp|_y  \left( \frac{|rp|_x}{|rp|_y} + \frac{|rq|_x}{|rq|_y}\right)}
    = 2 \frac {|rq|_y^2} {|rp|_y  \left(|rp|_x \cdot \frac{|rq|_y}{|rp|_y} + |rq|_x\right)}
    = 2 \frac{|rq|_y^2}{|rp|_y \cdot |sq|_x}
    \leq 2 / \myconst.
    \]

Lastly we assume that $\slope_{rp} < \slope_{rp}$ and $|rq|_x \geq |pq|_x / 2$. Here we use slightly different bounds on the area, namely that $\area(\trian_{i-1}) \leq |rp|_y \cdot |pq|_x$. We still have the bound of $\dist(\trian_{i-1}) \geq |rq|_y$ on the distance to the boundary, which gives us
    \[
    \slope_{i} - 2\slope_{i-1}
    = \slope_{rq} - 2\slope_{pq}
    = \slope_{rq} \left( 1 - 2 \frac{|pq|_y}{|pq|_x} \cdot \frac{|rq|_x}{|rq|_y} \right)
    \leq \slope_{rq} \left( 1 - \frac{|pq|_y}{|rq|_y} \right)
    \]\[
    \leq \slope_{rq} \left( 1 - \frac{|pq|_y}{|rp|_y} \right)
    = \slope_{rq} \frac{|rp|_y - |pq|_y}{|rp|_y}
    = \frac{|rq|_y^2}{|rq|_x \cdot |rp|_x}
    \leq 2 \frac{|rq|_y^2}{|pq|_x \cdot |rp|_x}
    %\leq 2 \frac{\dist(\trian_{i-1})^2}{\area(\trian_{i-1})}
    \leq 2 / \myconst.
    \]

\medskip
In each case we find that
    \[
    \slope_{i} \leq 2\slope_{i-1} + 2 / \myconst
    \leq 2 \cdot (2^{i} - 2) / \myconst + 2 / \myconst
    = (2^{i+1} - 2) / \myconst.
    \]
\end{proof}

\begin{lemma}
\label{lem:prop-iii}
If for all $j < i$ the edge $e_j$ is inside $\halfbox$ and $\max(4, \slope_j^2) \cdot |e_j| < \frac{\myconst}{8\sqrt{2}}$, then $|e_i| \leq 8 \eps \cdot 2^{(i+1)(i+7)} / \myconst^{3i}$.
\end{lemma}
\begin{figure}
\centering
\includegraphics{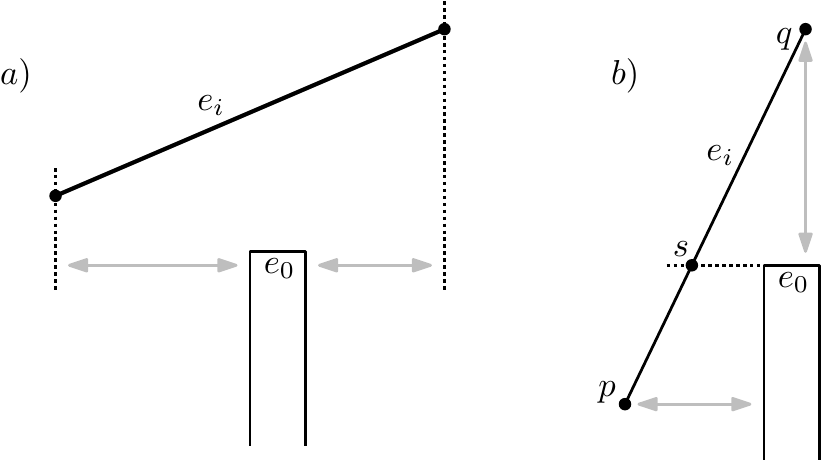}
\caption{Illustration of the two cases in the proof of Lemma~\ref{lem:prop-iii} based on the slope of $e_i$. Gray double arrows indicate distances used in the proof.}
\label{fig:illustration-lemma-7}
\end{figure}
\begin{proof}
Note that from Lemmas~\ref{lem:prop-i}~and~\ref{lem:prop-ii} we already know that $e_i$ is inside $\halfbox$ and $\slope_i \leq (2^{i+1} - 2) /\myconst$.
We first give bounds on the area and the distance to the boundary of $T_{i-1}$. We upper bound $\area(\trian_{i-1}) \leq |e_{i-1}| \cdot |e_{i}|$. For a lower bound on the distance to the boundary we look at the distance of points on $e_i$ to the boundary. From property (i) in Theorem~\ref{th:many-steiner} we know that $e_i$ is also inside $\halfbox$. Let $p$ and $q$ denote the two endpoints of $e_{i}$, without loss of generality we assume that $q$ is above or at the same height as $p$ and that $p$ is to the left of $q$. We distinguish two cases based on $\slope_{i}$.

First, the case when $\slope_{i} \leq 1$, as illustrated in Figure~\ref{fig:illustration-lemma-7}a. In this case $|e_{i}|_x \geq |e_{i}|/\sqrt{2}$ and $\dist(\trian_{i-1}) \geq (|e_{i}|/\sqrt{2} - \eps) / 2$. Filling this into our $\myconst$-distance property we get
\[
    1/\myconst \geq \frac{\dist(\trian_{i-1})^2}{\area(\trian_{i-1})}
    \geq \frac{(|e_i|/(2\sqrt{2}) - (\eps/2))^2}{|e_{i-1}| \cdot |e_{i}|}
    = \frac{|e_{i}|^2/8 - |e_{i}|\cdot \eps/(2\sqrt{2}) + \eps^2/4} {|e_{i-1}|\cdot|e_{i}|}
\]\[
    \geq \frac{|e_{i}|^2/8 - |e_{i}|\cdot \eps/(2\sqrt{2})} {|e_{i-1}|\cdot|e_{i}|}
    = \frac{|e_{i}|/8 - \eps/(2\sqrt{2})} {|e_{i-1}|},
\]
which can be rewritten as
\[|e_{i}| \leq 8/\myconst \cdot |e_{i-1}| + 4\eps / \sqrt{2}.\]

In the second case, we have $\slope_{i} > 1$, as shown in Figure~\ref{fig:illustration-lemma-7}b. Let $s$ be a point on the supporting line of $e_{i}$ that is horizontally aligned to $e_0$, then at least one of the edges $ps$ or $qs$ must have length at least $e_{i}/2$. If $|qs| \geq |e_{i}|/2$ we find that $\dist(\trian_{i-1}) \geq |e_{i}| / (2\sqrt{2})$. If this is not the case, then the edge $ps$ is a segment of $e_i$ and must be below $e_0$ and $|ps| \geq |e_{i}|/2$. Since $ps$ cannot intersect the boundary of $P$ and its nearest points are on one of the vertical neighbors of $e_0$ either $p$ or $s$ has a horizontal distance of at least $|e_{i}|/(2\sqrt{2}\slope_{i})$ to the boundary. We again fill this into our region property to get
\[
    1/\myconst \geq \frac{\dist(\trian_{i-1})^2}{\area(\trian_{i-1})}
    \geq \frac{(|e_{i}|/(2\sqrt{2}\slope_{i}))^2}{|e_{i-1}| \cdot |e_{i}|}
    = \frac{|e_{i}|^2}{8\slope_{i}^2 \cdot |e_{i-1}| \cdot |e_{i}|}
    = \frac{|e_{i}|}{8\slope_{i}^2 \cdot |e_{i-1}|}.
\]
Rewriting this we get
\[|e_{i}| \leq 8 / \myconst \cdot \slope_{i}^2 |e_{i-1}|.\]

Combining these two we find that
\begin{align*}
|e_{i}| &\leq 8 / \myconst \cdot (1+ \slope_{i}^2) |e_{i-1}| + 4\eps / \sqrt{2}
\\ &\leq  8 / \myconst \cdot (1 + ((2^{i+1} - 2)/\myconst)^2) |e_{i-1}| + 4\eps / \sqrt{2}
\\ &\leq 8 / \myconst \cdot (1 + 2^{2i+2} / \myconst^2) |e_{i-1}| + 4\eps / \sqrt{2}
\\ &\leq 2^{2i+6} / \myconst^3 |e_{i-1}| + 4\eps / \sqrt{2}
\\ &\leq 2^{2i+6} / \myconst^3 \cdot 8 \eps / \myconst^{3(i-1)} \cdot 2^{(i)(i+6)}  + 4\eps / \sqrt{2}
\\ &= 8 \eps / \myconst^{3i} \cdot 2^{(i+1)(i+7) - 1} + 4\eps / \sqrt{2}
\\ &\leq 8 \eps / \myconst^{3i} \cdot 2^{(i+1)(i+7)}.
\end{align*}
Note that in the last step we assume that $\myconst \leq 1$. Which is fine, since higher values of $\myconst$ only make the $\myconst$-distance property stricter, and we are constructing a lower bound.
\end{proof}

\medskip

Theorem~\ref{th:many-steiner} implies that we cannot restrict ourselves to triangulations
if we want a linear-size decomposition with the $\myconst$-distance property.
We hence consider possibly non-conforming decompositions (that is, we allow T-junctions) using convex $k$-gons.
We first show how to compute a linear-size decomposition with the $\myconst$-distance property
that uses convex $k$-gons for $k\leq 7$, and then we argue that each $k$-gon
can be further subdivided into convex quadrilaterals and triangles.

\begin{figure}[b]
\centering
\includegraphics{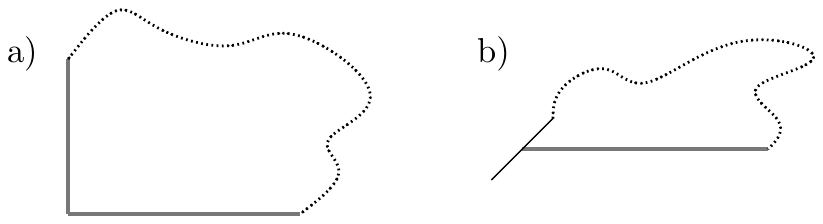}
\caption{Polygons on which we recurse consist of up to two subdivision edges and a boundary chain.}
\label{fig:recursion-polygon}
\end{figure}

\paragraph{A decomposition with 7-gons}
We assume without loss of generality that no two vertices of the input polygon~$P$ have the same $x$- or $y$-coordinates. We describe a recursive algorithm that computes in each step a single 7-gon\footnote{From now on, when we use the term 7-gon, we mean a convex $k$-gon for  $k\leq 7$.}
of the subdivision and then recurses on up to four smaller polygons. In a generic step of the recursive procedure, we are given a polygon
bounded by a chain of edges from the original polygon and by two subdivision edges,
one vertical and one horizontal; see Figure~\ref{fig:recursion-polygon}a.
(In our figures we use gray lines for subdivision edges, solid black lines for polygon edges, and dotted black lines to indicate an unspecified continuation of the boundary of the input polygon. Black disks mark vertices of the input polygon.)
The subdivision edges meet in a vertex, the \emph{corner} of the polygon.
One of the subdivision edges can have zero length (see Figure~\ref{fig:recursion-polygon}b). Without loss of generality we assume that
the horizontal subdivision edge, $e_\hor$, is the longer of the two subdivision edges,
and that the vertical subdivision edge $e_\ver$ extends
upward from the left endpoint of $e_\hor$.
Initially, $P$ does not have the right form as there are no subdivision edges. Hence we first pick an arbitrary point in the interior of $P$ and shoot axis-aligned rays in all four directions. This subdivides $P$ into four polygons that each have exactly two subdivision edges that meet in a vertex.

\begin{figure}
\centering
\includegraphics{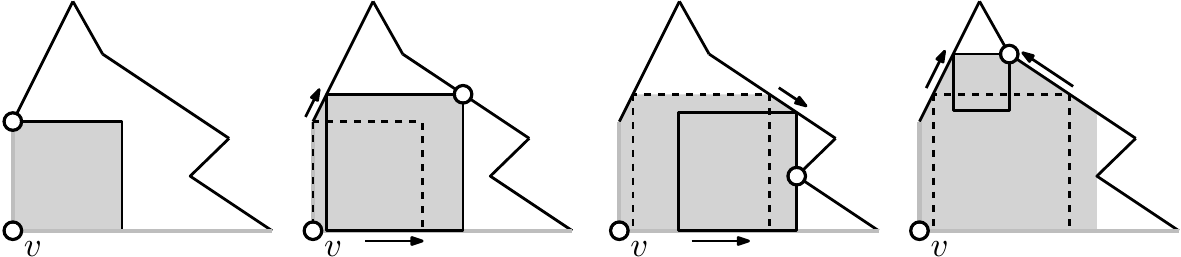}
\caption{Example of constructing a 7-gon from a recursion polygon.}
\label{fig:7-gon-example}
\end{figure}

We now describe how we generate a 7-gon of the decomposition in a recursive step
on input polygon~$Q\subset P$ with two subdivision edges, $e_h$ and $e_v$, meeting in  corner~$v$, see Figure~\ref{fig:7-gon-example} for an example.
We first grow a square with $v$ as lower-left corner, until the square
hits the boundary of~$Q$. (This could be immediately, if the vertical subdivision edge
has zero length.) If one of the edges of the square hits a vertex of the original polygon~$P$,
we stop. Otherwise a vertex of the square hits an edge of~$P$. We then start pushing
the square along the edge, meanwhile growing it so that it remains in contact with the subdivision edge. This again continues until the boundary of $P$ is hit, which may either terminate
the process (when a vertex of $P$ is hit) or not (when an edge is hit), and so on.
The 7-gon will be the union (swept volume) of all squares generated during the entire process. Figure~\ref{fig:cases-overview} gives an overview of the cases that can arise,
with A being the start configuration. Thick arrows indicate a transition from one case to another.
As mentioned, we stop pushing a square when a new vertex of $P$ occurs on the boundary. Cases where this happens are given a number (A1, B1, B2, $\ldots$). Next we provide more details on how to push the squares in each of the cases and when one case transitions to another. The top left, top right, bottom left, and bottom right vertex of a square will be denoted by $p_\nw, p_\NE, p_\sw, p_\se$, respectively, and the top, right, bottom, and left edge by $e_\north, e_\east, e_\south, e_\west$.
In each case the process ends when a vertex of $P$ is hit.

\myitemize{
\item[A] We grow a square from the corner while keeping $e_s$ on $e_h$ and $e_w$ on $e_v$ until it hits an edge or vertex of~$P$. We go into case~B if $p_\nw$ hits an edge $e_\nw$ of~$P$ or into case~E and~F if $p_\NE$ hits an edge $e_\NE$. Note that $p_\se$ cannot hit an edge of the polygon before $p_\nw$, since $e_h$ is at least as long as $e_v$.
\item[B] The vertex $p_\nw$ is on an edge $e_\nw$ of $P$ and $e_\south$ is on $e_\hor$. We push the square to the right while maintaining these contacts. We go into case~C if $p_\se$ hits an edge $e_\se$ of $P$ or into case~D and~F if $p_\NE$ hits an edge $e_\NE$ of $P$.
\item[C] The vertex $p_\nw$ is on an edge $e_\nw$ of $P$ and $p_\se$ is on an edge $e_\se$ of $P$. We push the square up and to the right maintaining these contacts. We go into case~D and~G if $p_\NE$ hits an edge $e_\NE$ of $P$.
\item[D] The vertex $p_\nw$ is on an edge $e_\nw$ and $p_\NE$ is on an edge $e_\NE$ of $P$. We push the square upward while maintaining these contacts.
\item[E] The vertex $p_\NE$ is on an edge $e_\NE$ of $P$ and $e_\west$ is on $e_\ver$. We push the square upward while maintaining these contacts. We go into case D if $p_\nw$ hits an edge of $P$.
\item[F] The vertex $p_\NE$ is on an edge $e_\NE$ of $P$ and $e_\south$ is on $e_\hor$ and we push the square to the right while maintaining these contacts. We go into case G if $p_\se$ hits an edge of $P$.
\item[G] The vertex $p_\NE$ is on an edge $e_\NE$ and $p_\se$ is on an edge $e_\se$ of $P$, and we push the square to the right while maintaining these contacts.
}

\begin{figure}
\centering
\includegraphics{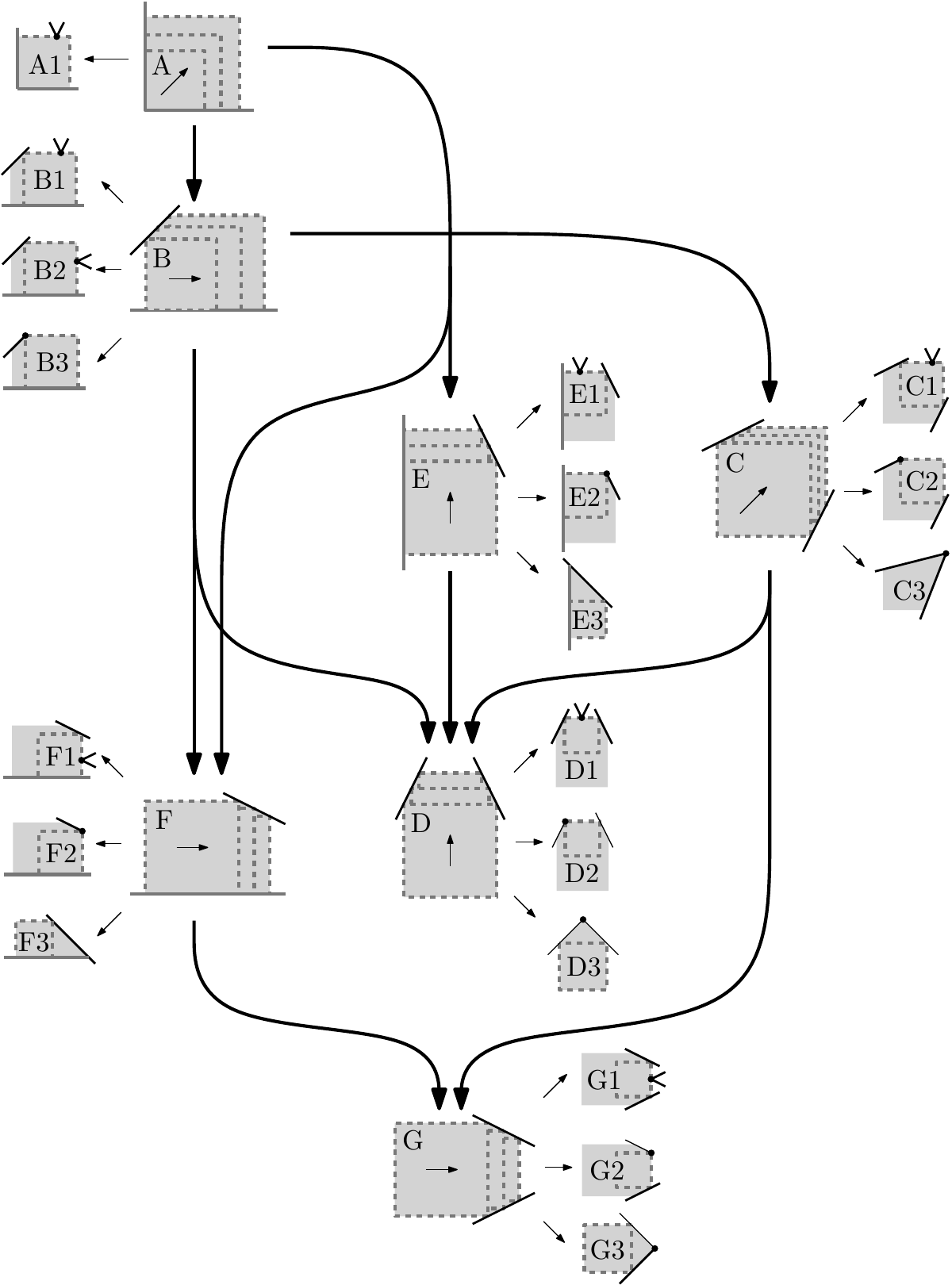}
\caption{We construct 7-gons by pushing squares through the polygon according to cases~A to~G. Fat arrows indicate a transition from one case to another and a split means that we continue in two separate directions. Note that cases E and F, and D and G are symmetric.}
\label{fig:cases-overview}
\end{figure}

\begin{lemma}
The process above generates a convex $k$-gon $C$ with $k\leq 7$. Moreover, for any
$p\in C$ we have $\area(C) \geq \frac{1}{2} \cdot \dtb{p}^2$, where
$\dtb{p}$ denotes the distance from $p$ to the boundary of the original polygon~$P$.
\end{lemma}
\begin{proof}
A straightforward case analysis of the different paths that the process may
follow in Figure~\ref{fig:cases-overview}---note that we can actually follow
several paths, since sometimes we continue pushing in two separate directions---shows
that $C$ is a convex 7-gon.
The construction guarantees that $C$ is the union of a (possibly infinite)
set of squares that each touch the boundary of~$P$. Let $p$ denote a point in $C$ and $\sigma$ a square containing $p$ that touches the boundary of $C$.
Then $\dtb{p} \leq \sqrt{2} \cdot \length(\sigma)$, where $\length(\sigma)$ denotes the edge length of $\sigma$.
It follows that
$\area(C) \geq \area(\sigma) = \length(\sigma)^2 \geq \frac{1}{2}\cdot \dtb{p}^2$.
\end{proof}

\begin{wrapfigure}[6]{r}{3cm}
\vspace{-10pt}
\raggedleft
\includegraphics{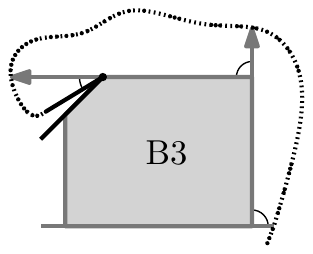}
\end{wrapfigure}

After constructing the 7-gon $C$, we should recurse on the remaining
parts of the polygon. The parts we can recurse on must have at most two orthogonal subdivision edges that meet in a point, as in Figure~\ref{fig:recursion-polygon}.
Parts for which this is not yet the case are first subdivided further by shooting horizontal and/or vertical rays from
certain vertices of $C$ so that the required property holds for the resulting subparts. Which rays to shoot depends on the final case in the construction
of $C$. The figure to the right shows case B3; the corners of the parts
on which we recurse are indicated by small circular arcs. In total, we may get up to four parts in which we recursively construct new 7-gons. Next, we bound the total number of regions that are created.

\begin{lemma}
\label{lem:region-count}
The algorithm described above creates $O(n_P)$ 7-gons in total,
when applied to a polygon $P$ with $n_P$ vertices.
\end{lemma}
\begin{proof}
Let $V_Q$ denote the subset of vertices of $P$ that are on the boundary of a
polygon $Q\subset P$ on which we recurse, excluding the possible vertices of $P$
that are the endpoints of the subdivision edges of $Q$. Recall that after we construct a
7-gon $C$ inside $Q$, the remainder of $Q$ is subdivided into at most four parts
on which we recurse again. At least one vertex of $V_Q$ is on the boundary of $C$,
so each part has strictly fewer vertices of $P$ on its boundary.
We also know that each vertex of $V_Q$ can be on the boundary of at most one
part (recall that vertices on endpoints of subdivision edges are not considered). It follows that only $O(n_P)$ 7-gons are constructed.
\end{proof}

Next we describe how to implement the algorithm in $O(n_P \log n_P)$ time. Each of the cases A to G can be viewed as moving a square from a start location to an end location such that all intermediate squares have specific contacts to the polygon $Q$ as detailed in each case description. To find the swept volume of this sequence of squares it suffices to know in each of the cases at which squares we start and end. To find these start and end squares we need some supporting data structures.

We use the medial axis $\medaxis$ of $P$, with the following asymmetric convex distance function. Let~$p$ and~$q$ be two points in the plane. The distance from $p$ to $q$ is the edge length of the smallest square with its lower left corner on $p$ that has $q$ on its boundary. This is different from the $L_\infty$ distance, since we grow a square from its corner, not its center. As a result, the ``distance'' from $p$ to $q$ is defined only if $q$ lies to the north-east of $p$. However,
for any point inside $P$ the distance to the boundary of $P$ and the nearest point on the boundary are well defined, which is sufficient for our purposes. Conceptually, one can also set the undefined distances to infinity.

Such a medial axis is the same as the Voronoi diagram of the line segments that form the polygon boundary with respect to a convex distance function. Fortune showed how to compute a Voronoi diagram of line segments in $O(n\log n)$ time for the Euclidean distance using a sweepline approach~\cite{f1987}. This approach can be extended to convex distance functions, even when the reference point is on the boundary as in our case~\cite{f1985,f2015}.
We then construct the following data structures:
\myitemize{
\item
    We preprocess each medial axis so that we can do point location in $O(\log n)$ time.
    Since the medial axis is a connected polygonal subdivision this can be done in $O(n)$ time~\cite{k1983}.
\item
    We also preprocess each medial axis so that we can answer horizontal and vertical
    ray shooting queries in $O(\log n)$ time. This can again be done in $O(n)$ time by first computing the horizontal and vertical
    decomposition of $P$~\cite{c1991}, and then preprocessing these trapezoidal maps for point location.
\item Finally, we preprocess $P$ itself in $O(n)$ time such that we can do horizontal and vertical ray shooting in $O(\log n)$ time.
}

Initially (case A) we want to find the largest square that we can grow from the corner~$v$. We locate the cell of $\medaxis$ that contains $v$, which gives us the vertex or edge of the polygon, say edge $e$, that is closest to $v$ in the specified distance measure. This implies that $e$ is the first edge hit by the boundary of a square grown from $v$.
In this way we determine in $O(\log n)$ time if we are in case~A,~B, or~E. Next we push the square upward or to the right. We then have to determine the final square for that movement and in which case we should continue. We distinguish two different types of movement for the square. Either the square has one edge on one of the vertical or horizontal subdivision edges (case~B,~E, and~F), or it has two corners on polygon edges of~$P$ (case~C,~D, and~G).

If one edge of the square stays on a subdivision edge then specifically the lower left vertex stays on the subdivision edge and the series of squares that we create are exactly the largest squares with their lower left corners on the subdivision edge. Recall that we stop moving the square when another edge or vertex of $P$ hits the boundary of the square. Let $q$ denote the lower left corner of this square. By definition of $\medaxis$ the point $q$ has to be on a bisector of $\medaxis$ as there are two different features (edges or vertices) of $P$ that are at equal distance. Hence, the process of moving a square along a subdivision edge is essentially the same as moving its lower left corner point until it hits an edge of the medial axis (or $P$). We can use horizontal or vertical ray shooting to find in $O(\log n)$ time the point $q$ where we end the movement along the subdivision edge.

When we move a square while keeping two vertices on edges of $P$ it follows from the definition of the medial axis and our distance measure that the lower left vertex of the square remains on the bisector of the two edges of $P$. The movement ends when a third edge or vertex of $P$ is on the boundary of the square, so at a vertex of the medial axis. Specifically the vertex where the bisector along which the lower left vertex was moving, ends. To find the final square of the movement we have to find the bisector, determine which endpoint of the bisector we need and find the three edges or vertices of $P$ that define that vertex. Since we already found the right bisector in the previous case, each of these steps can be done in $O(1)$ time after which we can determine in $O(1)$ time how to continue.
To summarize, we obtain the following lemma.

\begin{lemma}
\label{ref:compute-7-gon}
Computing the 7-gon in a recursive step of the algorithm takes $O(\log n)$ time,
after $O(n)$ preprocessing.
\end{lemma}

\paragraph{From 7-gons to quadrilaterals and triangles}
As a last step we can convert the 7-gons from our decomposition
into convex quadrilaterals and triangles. The resulting decomposition still
has the $\myconst$-distance property, although the value for $\alpha$
will decrease from 1/2 to 1/8, as shown below.
Let $Q$ denote a convex polygon with $n_Q$ vertices. By the ham-sandwich theorem~\cite{e1987}, there exists a line cutting~$Q$ into two portions of equal area with at most $\lfloor n_Q/2 \rfloor$ vertices of $Q$ strictly on each side of the line.  Cutting along this line, we obtain two polygons with half the area and at most $\lfloor n_Q/2 \rfloor + 2$ vertices each. By repeating this process, if necessary, we obtain either triangles or quadrilaterals. Using these ham-sandwich cuts we prove the following theorem.

\begin{theorem}
\label{thm:polygon-decomp}
Given a simple polygon $P$ we can compute in $O(n\log n)$ time a subdivision of $P$
consisting of $O(n)$ triangles and convex quadrilaterals with the $(1/8)$-distance property.
\end{theorem}
\begin{proof}
By Lemmas~\ref{lem:region-count}~and~\ref{ref:compute-7-gon} we can compute
in $O(n\log n)$ time a decomposition of $P$ into $O(n)$ convex $k$-gons,
for $k\leq 7$, that has the $(1/2)$-distance property.
We further subdivide each $k$-gon using ham-sandwich cuts, as explained above. In the worst case we start with a 7-gon that is split into two 5-gons by the first ham-sandwich cut, after which each 5-gon is split into two quadrilaterals. We then get four quadrilaterals each having 1/4 of the area of the 7-gon.
Thus, since the decomposition into 7-gons had the (1/2)-distance property, the new decomposition has the (1/8)-distance property.
\end{proof}

Combining this result with Theorem~\ref{thm:decomp-pl} we obtain the following corollary.

\begin{cor}
Let $\subdiv$ denote a planar polygonal subdivision with $O(n)$ vertices and let $\gamma_i$ for each $P_i \in \subdiv$ denote the probability that a query point lies in $P_i$. We can construct in $O(n\log n)$ expected time a point location structure that uses $O(n)$ space
and answers a query with a point $p$ in $O\left(\min \left(\log n, 1 + \log \frac{\area(P_i)}{\gamma_i\dtb{p}^2}\right)\right)$ time, where $\dtb{p}$ denotes the Euclidean distance from $p$ to the nearest point on any edge of $\subdiv$.
\end{cor}

\section{Depth-bounded quadtree}
Although computing a distance-sensitive decomposition takes $O(n\log n)$ time asymptotically there is a lot of overhead involved. During preprocessing we need several medial axes of the input subdivision $\subdiv$, and each of these has to be further processed for point location and horizontal and vertical ray-shooting. We also create many additional regions which would cause the worst-case $O(\log n)$ search time to have a much larger constant when compared to a worst-case optimal point location structure. In this section we present a much simpler solution that has very little extra overhead compared to a general worst-case optimal point locations structure, but only works for a special case of the problem.

In this special case we assume no distribution of the queries over the polygons of the subdivision is given and we want the query time to be dependant only on the distance from a point to the boundary. Let $\subdiv$ be a planar polygonal subdivision and assume that $\subdiv$ is contained in a square with area~1. (Note that in this case we do not require $\subdiv$ to be connected.) We show how to construct a query structure that can answer a query for a point $p$ in $O\left(\min \left(\log n, 1 + \log \frac{1}{\dtb{p}^2}\right)\right)$ time, where $\dtb{p}$ again denotes the Euclidean distance from $p$ to the nearest point on any edge of $\subdiv$. This can be seen as a special case of the general problem where each region $P_i \in \subdiv$ has a weight proportional to its area, so $\gamma_i = \area(P_i) / \area(\subdiv)$.

In essence we have two different requirements for a query. First, no query should ever take more than $O(\log n)$ time, and second, a query for a point far from the boundary should take only $O(1 + \log \frac{1}{\dtb{p}^2})$ time. A worst-case optimal point location structure can be used to satisfy the first requirement and a quadtree where each leaf intersects $O(1)$ features of the subdivision satisfies the second requirement. Unfortunately, neither satisfies both: a quadtree may have nodes with a very high depth and a worst-case optimal point location structure gives no guarantees on finding points far from the boundary quickly. We can however use both structures together to get the bound we need.

We construct two structures: a general worst-case optimal point-location structure $\pl(\subdiv)$ and a depth-bounded quadtree $\qt(\subdiv)$. With a slight abuse of terminology we use \emph{leaf, root} and \emph{node} to denote nodes of the quadtree as well as the square regions they are associated with. The root of the quadtree is the bounding square of $\subdiv$, which we assume to have edge length~1. Each leaf of the quadtree is either empty---it does not intersect the boundary of $\subdiv$---or it has a depth of $\lceil \log \sqrt{n} \rceil$; see Figure~\ref{fig:qt}a. A query for a point $p$ first finds the leaf $v$ of the quadtree that contains $p$. If $v$ does not intersect any of the boundary elements of $\subdiv$, then the polygon $P \in \subdiv$ that contains $v$ also contains $p$. If $v$ is not empty, then we conclude that $p$ is close to the boundary of $\subdiv$ and perform a query in $\pl(\subdiv)$.

\begin{figure}
\centering
\includegraphics{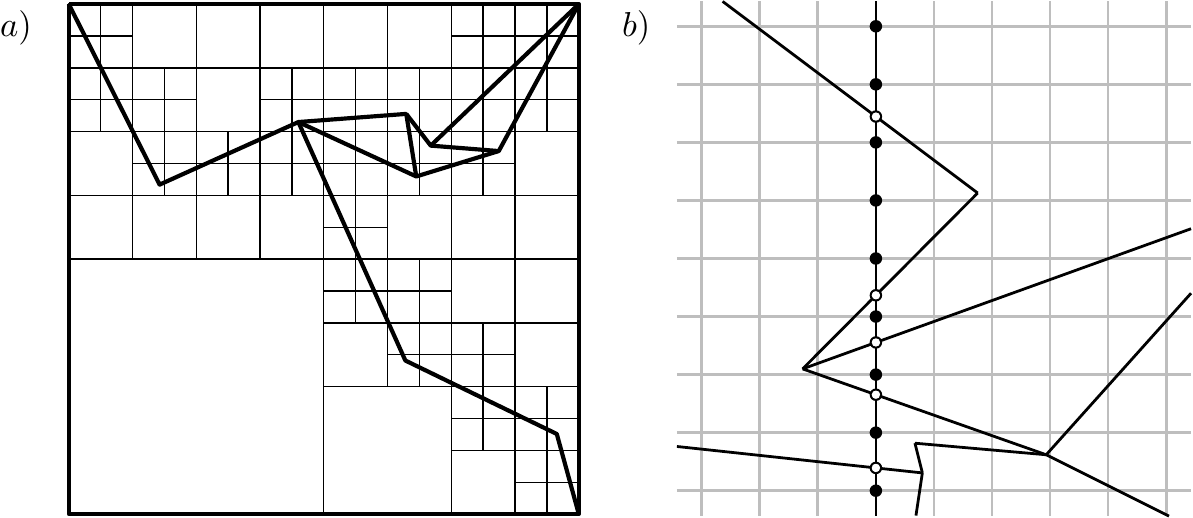}
\caption{$a)$ A depth-bounded quadtree $\qt(\subdiv)$. $b)$ An illustration of the sweep-line algorithm. Closed disks indicate grid vertices and open disks indicate intersection points of subdivision edges with the sweep-line.}
\label{fig:qt}
\end{figure}

\paragraph{Preprocessing} Constructing a worst-case optimal point-location structure takes $O(n \log n)$ time, where~$n$ is the complexity of $\subdiv$.  When constructing the quadtree we have to account for the presence of edges of $\subdiv$, and not just its vertices. The standard method to construct a quadtree on a set of points is to recursively split nodes that contain more than one point and propagate the points down the tree such that each leaf stores the points contained in its associated square. In our case each leaf would have to store the edges that intersect it, which would lead to superlinear storage as each edge may intersect many leaves of the quadtree. Instead we use a different approach that uses a sweep-line over the underlying grid of the quadtree.

We first construct the complete quadtree up to depth $\lceil \log \sqrt{n} \rceil$, which represents a grid where each cell has an edge length $\ell$ between $1 / (2\sqrt{n})$ and $1 / \sqrt{n}$. It follows that the grid contains $O(n)$ cells in total. We will mark each leaf of the quadtree whose associated grid-cell is intersected by an edge of~$\subdiv$. A cell of the grid is intersected by an edge of $\subdiv$ if and only if either one of its boundary segments intersects an edge or if the cell contains a vertex of $\subdiv$. We can mark leaves that contain a vertex by locating each vertex within the grid, which takes $O(n \log n)$ time. We then use two sweep-lines to mark cells whose boundary segments are intersected by edges of $\subdiv$. We use a horizontal sweep-line to mark all leaves whose grid cells have their left or right boundary segment intersected by an edge of $\subdiv$. The sweep goes from left to right and we maintain an ordered list of edges from the subdivision that intersect the sweep-line. This ordering changes only when the sweep-line encounters vertices of the subdivision. When the sweep-line encounters a vertex $v$ we locate the vertex in the current edge-ordering in $O(\log n)$ time and then spend $O(k \log n)$ time adding and removing edges adjacent to $v$, where $k$ is the degree of $v$. As there are $O(n)$ vertices and the sum of their degrees is $O(n)$ the vertex events take $O(n\log n)$ time in total. When the sweep-line encounters a vertical line of the grid we test for intersections between the subdivision edges stored in the sweep line and the vertical grid-segments that coincide with the sweep line. Each vertical grid-segment---the boundary edge of one or two cells---is intersected if and only if there are edges of $\subdiv$ between its endpoints on the sweep-line. This is easy to test by simply locating each grid-vertex on the vertical line in the edge-ordering stored in the sweep-line; see Figure~\ref{fig:qt}b. If a grid-segment is intersected by an edge of the subdivision we mark the leaves whose cells are to the left and right of this grid-segment. For each such event we have to perform $O(\sqrt{n})$ binary searches on the edge-ordering of the sweep-line, taking $O(\sqrt{n} \log n)$ time in total. Since there are $O(\sqrt{n})$ such events this takes $O(n\log n)$ time in total. This sweep marks all cells of which the left or right boundary edge is intersected by a subdivision edge. A similar vertical sweep is used to mark all leaves of which the top or bottom segment of its associated grid-cell is intersected by a subdivision edge.

After performing both sweeps each leaf intersected by the subdivision boundary is marked. Next we mark internal nodes of the quadtree of which the associated square intersects an edge of $\subdiv$. We use a bottom-up approach where each node is marked if and only if at least one of its children is marked. Next, the tree is trimmed by removing all nodes with an unmarked parent. The resulting quadtree is a depth-bounded tree in which each leaf has depth $\lceil \log \sqrt{n} \rceil$ or does not intersect the boundary of the subdivision $\subdiv$. As a final step we do a single point location for each empty (not marked) leaf of the quadtree to determine which polygon of $\subdiv$ it is contained in and store this information in the leaf.

\begin{lemma}\label{lem:2d-prep}
Given a subdivision $\subdiv$, we can construct the depth-bounded quadtree $\qt(\subdiv)$ and worst-case optimal point-location structure $\pl(\subdiv)$ in $O(n\log n)$ time, where $n$ is the complexity of $\subdiv$.
\end{lemma}

\paragraph{Querying} Given the quadtree $\qt(\subdiv)$ and the point location structure $\pl(\subdiv)$ we perform a point location query on a point $p$ as follows. We first find the leaf $v$ of $\qt(\subdiv)$ that contains $p$. If $v$ is empty, then we report the polygon that contains $v$, otherwise we do a point location query for $p$ in $\pl(\subdiv)$ to find the polygon containing $p$. Next we show that this indeed provides us with the required query-time.

\begin{lemma}\label{lem:2d-query}
A point-location query as described above for a point $p$ takes $O(\min(1 + \log \frac{1}{\dtb{p}^2}, \log n))$ time, where $\dtb{p}$ is the distance from $p$ to the boundary of $\subdiv$.
\end{lemma}
\begin{proof}
We distinguish two cases. First assume the leaf $v$ from $\qt(\subdiv)$ that contains $p$ is empty. Let $i$ denote the depth of $v$ in the quadtree, so we spend $O(i)$ time to locate $p$. The node $v$ has an edge length of $1/2^i$, and its parent has an edge length of $2/2^i$. The parent of $v$ was split, so it must have intersected the boundary of $\subdiv$. This implies that $\dtb{p} \leq 2\sqrt{2} / 2^i$, since both $p$ and some point on the boundary of $\subdiv$ are contained in the parent of $v$. Plugging this in, we find that indeed
\[
O(i) = O\left(\min\left(1 + \log \frac{1}{2\sqrt{2} / 2^i}, \log n\right)\right) = O\left(\min\left(1 + \log \frac{1}{\dtb{p}^2}, \log n\right)\right).
\]

Now suppose $v$ is not empty. In this case we spend $O(\log n)$ time in the quadtree and $O(\log n)$ time in the general point location structure. However, since $v$ must have an edge length of at most $1/\sqrt{n}$ and is intersected by the boundary of $\subdiv$ we know that $\dtb{p} \leq \sqrt{2} / \sqrt{n}$ and the query bound follows.
\end{proof}

Combining Lemmas~\ref{lem:2d-prep}~and~\ref{lem:2d-query}, we obtain the desired result.

\begin{theorem}
Given a planar piecewise-linear subdivision $\subdiv$ contained in a square with edge length~$1$, we can construct in  $O(n\log n)$ expected time a point location structure that can answer a query for a point $p$ in $\subdiv$ in $O(\min(1 + \log \frac{1}{\dtb{p}^2}, \log n))$ time, where $\dtb{p}$ denotes the distance from $p$ to the boundary of the polygon $P\in \subdiv$ that contains it.
\end{theorem}

\paragraph{Convex subdivisions in $\Reals^3$}
The above method of using a depth-bounded quadtree together with a worst-case optimal point-location structure can also be applied to convex subdivisions in~$\Reals^3$. In this case we would want to compute a depth-bounded octree, where each leaf either does not intersect any boundary facet or has depth $\lceil \log \sqrt[3]{n} \rceil$. As before we can first construct the full octree of depth $\lceil \log \sqrt[3]{n} \rceil$ and then mark leaves that intersect the subdivision boundary. In a general connected subdivision in 3D a cell is intersected if and only if its 2-dimensional faces are intersected by a subdivision facet. The straightforward extension of the sweep-line approach from the 2-dimensional case would require us to maintain a dynamic subdivision defined by the intersection of the input subdivision $\subdiv$ and the sweep-plane. Then whenever the sweep-plane encounters a plane in the grid we should determine if the boundary squares of the grids cells are empty in the sweep-plane. This seems difficult to do in near-linear time, as we cannot afford to traverse the entire sweep-plane, which may have $\Theta(n)$ complexity. However, in a convex subdivision a grid cell is intersected by a subdivision facet if and only if at least two of its vertices are in different cells of the subdivision. As a result we can simply perform a point location query on each vertex of the grid and test for each grid cell whether all vertices are contained in the same polyhedron of the subdivision. If not all vertices belong to the same polyhedron we mark the associated leaf of the octree. We can use the $O(n\log n)$ space structure by Snoeyink~\cite{s2004} to perform each query in $O(\log^2 n)$ time. After marking the leaves of the octree we propagate the marking upwards, trim the tree and determine for each empty leaf which polyhedron contains it, similar to the two-dimensional case. A query for a point $p$ is again performed by first locating $p$ in the octree, where at most $O(\log n)$ time is spent. If the resulting leaf is not empty we instead find $p$ in the general point location structure in $O(\log^2 n)$ time.

\begin{theorem}
Given a 3-dimensional convex polyhedral subdivision $\subdiv$ contained in a cube with edge length~$1$, we can construct in $O(n \log^2 n)$ time and $O(n\log n)$ space a point location structure that can answer a query for a point $p$ in $\subdiv$ in $O(\log \frac{1}{\dtb{p}^2})$ time if $\dtb{p} \geq \sqrt{3} / \sqrt[3]{n}$ and $O(\log^2 n)$ otherwise, where $\dtb{p}$ is the shortest distance from $p$ to nearest boundary facet of $\subdiv$.
\end{theorem}

\section{Conclusions}
We presented two data structures for distance-sensitive point location. The first and most general structure relies on decomposing a connected planar subdivision into constant complexity regions, such that any point that is far from the boundary is contained in a large region. We then showed how such a distance-sensitive decomposition is used to create a distance-sensitive point location structure. Computing the decomposition and the point location structure takes $O(n\log n)$ time and $O(n)$ space. A query for a point $p$ with distance $\dtb{p}$ to the nearest point on an edge of $\subdiv$ takes $O\left(\min \left(\log n, 1+ \log \frac{\area(P_i)}{\gamma_i\dtb{p}^2}\right)\right)$ time, where $\gamma_i$ denotes the given probability that a query falls in polygon $P_i$. Our distance-sensitive decomposition consists of triangles and quadrilaterals and may be non-conforming, that is, there may be T-junctions along their boundaries. An obvious questions is whether ``nicer'' decompositions are possible that have the same property that a point far from the boundary is guaranteed to be in a large region. We showed that if we insist on a conformal Steiner triangulation, then we cannot bound the number of regions as a function of $n$, the number of edges of $\subdiv$. For non-conformal triangulations or conformal quadrilaterals this questions is still open. Another interesting open question is if a similar decomposition is possible for subdivisions in three dimensions. Note, however, that this would not directly lead to distance-sensitive point location structure since, to our knowledge, no three-dimensional entropy-based point location structures are known.

We also presented a simpler structure that does not take into account the query distribution between different regions of the input subdivision $\subdiv$. Instead only the distance from a query point to the nearest edges of the subdivision is considered. This can be seen as a special case of the general distance-sensitive problem, where each each polygon $P_i$ has a probability $\gamma_i = \area(P_i)$. The point-location structure consists of a quadtree with a maximum depth of $\lceil \log \sqrt{n} \rceil$ and a general worst-case optimal point location structure, both of which can be constructed in $O(n \log n)$ time and $O(n)$ space. A query for a point $p$ then takes $O(\min(\log n, 1 + \log \frac{1}{\dtb{p}^2}))$ time. This is not asymptotically better than if we would use the general solution, but we believe this second structure is much simpler to construct and has a smaller overhead. The quadtree-based structure can also be extended to work for convex subdivisions in three dimensions. It takes $O(n\log^2 n)$ time and $O(n\log n)$ space to construct a worst-case efficient structure and a depth-bounded octree. A query then takes $O(1 + \log \frac{1}{\dtb{p}^2})$ time if $\dtb{p} \geq \sqrt{3} / \sqrt[3]{n}$ and $O(\log^2 n)$ time otherwise. Note that the $O(n\log n)$ space requirement comes from the worst-case efficient point location structure as no $O(n)$ space structure is yet known that has $O(\log^2 n)$ query time.

%-----------------------------------------------------------------------------------

\end{document}